\documentclass[11pt]{article}
\usepackage[T1]{fontenc}
\usepackage[hyperindex,breaklinks,colorlinks,citecolor=blue]{hyperref}
\usepackage{amsmath, amsfonts, amssymb, amsthm}
\usepackage{color}
\usepackage{latexsym}
\usepackage{fullpage}

\theoremstyle{plain}
\newtheorem{lemma}{Lemma}[section]
\newtheorem{proposition}{Proposition}[section]
\newtheorem{theorem}{Theorem}[section]
\newtheorem{corollary}{Corollary}[section]

\theoremstyle{definition}
\newtheorem{definition}{Definition}[section]

\theoremstyle{remark}
\newtheorem{remark}{Remark}[section]
\newtheorem{example}{Example}[section]

\newcommand{\prog}{Markov-}
\newcommand{\K}{K-}
\newcommand{\oracle}{Type-2-}

\newcommand{\progs}{Markov}
\newcommand{\Ks}{K}
\newcommand{\oracles}{Type-2}

\newcommand{\defin}[1]{\emph{\textbf{#1}}}
\newcommand{\bd}[1]{\boldsymbol{#1}}

\newcommand{\A}{\mathcal{A}}
\newcommand{\BB}{\mathbb{B}}
\newcommand{\B}{\mathcal{B}}

\newcommand{\N}{\mathbb{N}}

\renewcommand{\O}{\mathcal{O}}
\renewcommand{\P}{\mathcal{P}}

\renewcommand{\S}{\mathcal{S}}
\newcommand{\sierp}{\mathbb{S}}

\newcommand{\down}{\mathop{\!\downarrow}}
\newcommand{\dv}{\mathop{\!\uparrow}}
\newcommand{\at}{\text{at }}
\newcommand{\up}[1]{[#1]}

\newcommand{\cantor}{\{0,1\}^\N}

\newcommand{\Nbar}{\overline{\N}}

\newcommand{\dom}{\text{dom}}

\newcommand{\Km}{Km}

\newcommand{\U}{\mathcal{U}}
\newcommand{\V}{\mathcal{V}}
\newcommand{\W}{P_c(\N)}
\newcommand{\X}{\mathcal{X}}
\newcommand{\Y}{\mathcal{Y}}

\newcommand{\res}[1]{\mathbin\upharpoonright_{#1}}

\title{On the information carried by programs\\ about the objects they compute\footnote{This work was partially supported by Inria program ``Chercheur invit\'e''. CR was partially supported by FONDECYT project  11110226 and BASAL PFB-03 CMM, Universidad de Chile.}}

\author{Mathieu Hoyrup and Crist\'obal Rojas}
\begin{document}

\maketitle

\begin{abstract}
In computability theory and computable analysis, finite programs can compute infinite objects. Presenting a computable object via any program for it, provides at least as much information as presenting the object itself, written on an infinite tape. What additional information do programs provide? We characterize this additional information to be any upper bound on the Kolmogorov complexity of the object.  Hence we identify the exact relationship between \prog computability and \oracle computability. We then use this relationship to obtain several results characterizing the computational and topological structure of \prog semidecidable sets.  
 \end{abstract}

\section{Introduction}

We assume that the reader is familiar with Turing machines and basic computability theory  over the natural numbers. To define computability over infinite objects, one still uses Turing Machines but has to set up a way for them to access such objects. In any case, the input of the machine is a finite or infinite sequence of symbols written on the input tape and one has to choose a suitable way to describe infinite objects by such symbolic sequences. We now briefly describe the two main approaches that have been developed.  

The first one was introduced and studied by Turing \cite{Tur36}, Grzegorczyk \cite{Grzegorczyk57}, Lacombe \cite{Lacombe55} and later Kreitz and Weihrauch \cite{Wei00} and is nowadays known as Type-Two computability. In this model, the description itself is completely written on the input tape of the machine. At any time, the machine can read a finite portion of this description. We will call this the \defin{\oracle model}. The second approach, promoted by the Russian school led by Markov \cite{Markov54,Kushner06}, gives an alternative. In this model one restricts the action of the machine to operate on computable (infinite) objects only, in the sense that they have  computable descriptions. Instead of having access to description themselves as in the \oracle model, the machine here has access to a \emph{program} computing a description. We will call this the \defin{\prog model}. These two approaches provide a priori different computability notions, and their comparison has been an important subject of study \cite{Rice53, MyhillShe55, Shapiro56, KreLacSch57, Ceitin62, Friedberg58, PourEl60, Hertling96, Spreen01}.

 It is clear that the \prog model  is at least as powerful as the \oracle model, so the question is:  does it allow to compute  strictly more than the \oracle model?  The answer depends on the objects that we consider, and the algorithmic tasks we want to perform on them. The computational  power of these models can therefore be classified according to these parameters. Table \ref{ShortTable} summarizes the most celebrated results in this direction. The computable objects considered are the partial computable functions  and the total computable functions. The algorithmic tasks considered are decidability and semidecidability of properties about these objects.

 \begin{table}[htdp]
\caption{Some celebrated results comparing \prog computability to \oracle computability. }
\begin{center}\label{ShortTable}
\begin{tabular}{|r|c|c|}\hline
Objects \, \qquad \qquad \qquad & Decidability & Semidecidability \\ \hline
Partial computable functions &  $\underset{Rice}{\text{\progs}\equiv\text{\oracles}}$ & $\underset{Rice-Shapiro}{\text{\progs}\equiv\text{\oracles}}$ \\ \hline
Total  computable functions  &  $\underset{Kreisel\,\,et\,\, al/Ceitin}{\text{\progs}\equiv\text{\oracles}}$ &  $\underset{Friedberg}{\text{\progs}>\text{\oracles}}$ \\ \hline

\end{tabular}
\end{center}
\label{default}
\end{table}%

Kreisel-Lacombe-Sh\oe nfield/Ceitin's Theorem \cite{KreLacSch57, Ceitin62} for instance, states that over total computable functions,  \prog decidability   is equivalent to  \oracle decidability\footnote{In its original form, this theorem is stated for functionals.}. This means that the  machine trying to decide a property,  when provided with a program $p$ for a function $f$, cannot do better than just running $p$ to evaluate $f$. The machine gains no additional information about $f$ from $p$. We note that Ceitin's version of this result shows that over the real line, \prog computable functions and \oracle computable functions coincide. 

On the other  hand, Friedberg \cite{Friedberg58} exhibited properties about total computable functions that are \prog semidecidable but not \oracle semidecidable. So that for semidecidability,  a program $p$ for a function $f$ \emph{does} give some additional information that can be exploited by the machine. The main question we rise in this paper is the following:

\begin{center}\emph{Can we characterize the additional useful information contained in a program computing an object, as compared to having the object itself ? }
\end{center}

To get some intuition, consider the following fundamental difference between the two models.  In the \oracle model, at any given time only a  finite portion of the description of $x$ is provided,  which corresponds to a finite approximation of $x$. Clearly, this approximation is also good for  infinitely many other objects -- all the ones that are ``close enough'' to $x$. In particular, \emph{$x$ is never completely specified}. In the \prog model on the other hand, the program provided to the machine completely specifies $x$ from the beginning of the calculation!   This increases the predictive power of $M$, which might therefore be able to perform stronger calculations. The point is to understand in which situations this fact can be exploited. A trivial example is obtained when one considers the relativized setting:  \emph{every} function is \prog computable relative to an appropriate (powerful enough) oracle.  Whereas whatever oracle $A$ we consider, \oracle computable functions relative to $A$ must always be continuous.  

This observation takes us to another interesting point that separates the \prog model from the \oracle model, namely their topological structure. It is well known that \oracle computability and topology are closely related: e.g. the \oracle computable functions are exactly the effectively continuous ones, and the \oracle semidecidable properties exactly correspond to the effectively open sets.  The connection between \prog computability and topology, on the other hand, appears to be much less clear. In particular, Friedberg's construction provides a \prog semidecidable set which is not open (for the standard topology restricted to computable elements).

An obvious solution to relate \prog computability to topology is to consider precisely the topology generated by all the \prog semidecidable sets -- the so called Ershov's topology. The question then becomes: 

\begin{center}
How do \prog semidecidable sets look like? can we characterize Ershov's open sets ? 
\end{center}

In the present paper we  make use of Kolmogorov Complexity to provide a fairly complete  answer to these and other questions in different settings. Our main result is a characterization of the additional information provided by a program, when the class of objects considered  are the computable points of an effective topological space.  It can be informally stated as follows (see Section \ref{main-results}):

\bigskip
\noindent \textbf{Theorem A }\emph{Over effective topological spaces,  a program computing  $x$ provides as much information as a description  of $x$ itself  together with  \textbf{any upper bound on the Kolmogorov complexity of x.}}

\bigskip

Here, the Kolmogorov complexity $K(x)$ of a computable infinite object $x$ is to be understood as the size of the shortest program computing a description of $x$. Obviously, any program for $x$ trivially provides, in addition to a description,  an upper bound on its Kolmogorov complexity.  Theorem A  says that this bound is all the exploitable additional information it provides. 

Thus, we have a third model to deal with computable infinite objects. In this model, input $x$ is presented to the machine as a pair $(d,k)$, where $d$ is a description for $x$ and $k$ a bound on the Kolmogorov complexity of $x$. We shall call this the \defin{\K model.} In these terms, a particular case of Theorem {A} can be stated as follows:  if $\X$, $\Y$ are \emph{effective topological spaces} (not necessarily metric) and $X_c$, $Y_{c}$ are the corresponding set of computable points, then a function $f:X_c\to Y_{c}$ is \prog computable if and only if it is \K computable.

A simple observation shows that one can not in general compute a program for $x$ from a \K description of $x$, meaning that the two notions are not fully equivalent. Despite this fact,  Theorem {A} is valid in great generality:  it holds for decidability, semidecidability and also higher in the hierarchy. In proving this we make a fundamental use of the Recursion Theorem.  Interestingly, although the Recursion Theorem does not relativize (a well known fact),  Theorem {A} relative to the halting set also holds in many cases,  but for entirely different reasons.

The \K model also sheds light into the structure of the open sets of  Ershov's topology, providing  a nice characterization in terms of Kolmogorov complexity, at least in the  particular case of the extended natural numbers $\overline{\N}=\N \cup \{\infty\}$.

\medskip

\noindent \textbf{Theorem B }\emph{On the extended natural numbers, the Ershov topology is generated by the sets $\{n\}$ and $\{x\in\overline{\N}:K(x)<h(x)\}$ for some computable order $h$.}

\medskip

With the same techniques, we are able to prove several other related results that are interesting on their own. For example, we show that there is no effective enumeration of the \prog semidecidable sets of $\overline{\N}$ and that there is a \prog semidecidable subset of $\{0,1\}^{\N}$ that is not $\Sigma_{2}^{0}$.

Finally, in the search of the limitations of our techniques, we turn our attention to more general spaces and analyze functions with values on topological spaces that have an admissible representation but are not countably-based. In particular, when this is the space of open subsets of Baire space $\O(\BB)$, we show that \prog computability can be strictly stronger than \K computability:

\medskip

\noindent\textbf{Theorem C}. \emph{For functions from the partial computable functions with values on $\O(\BB)$ one has that:}
\begin{equation*}
\emph{\prog computability \, >\, \K computability \, > \, \oracle computability.} 
\end{equation*}

One of the main question that remains open is whether the first strict inequality in Theorem C holds if we replace the \emph{partial} computable functions by the \emph{total} ones. The situation is summarized in table \ref{LongTable}.

\begin{table}[htdp]
\caption{Some results comparing \prog computability, \K computability and \oracle computability. $\sierp=\{\bot,\top\}$ is the Sierpi\'nski space whose topology is generated by $\{\top\}$.}
\begin{center}\label{LongTable}
\begin{tabular}{|c|c|c|c|}\hline
Space  $\X$ & Semidecidable & $\emptyset '$-Semidecidable & F$: \X \to \O(\BB)$\\ \hline
$\sierp$ &  \progs{} ${\equiv}$ \Ks{}  $\equiv$ \oracles{} & \progs{} ${>}$ \Ks{} $\equiv$ \oracles{} & \progs{} $>$ \Ks{} $\equiv$ \oracles{} \\ \hline
Partial functions &  \progs{} $\equiv$ \Ks{} $\equiv$ \oracles{} & \progs{} > \Ks{} $\equiv$ \oracles{}  & \progs{} > \Ks{} > \oracles{} \\ \hline
Total functions & \progs{} $\equiv$  \Ks{} > \oracles{}  & \progs{} $\equiv$ \Ks{} > \oracles{}  &  \progs{} \textcolor{red}{?} \Ks{} > \oracles{} \\ \hline
\end{tabular}
\end{center}
\label{default2}
\end{table}%

The paper is organized as follows. We start by providing the basic notions and definitions in Section \ref{sec_background}. In Section \ref{main-results} we introduce the \K model and present our main results.  Section \ref{sec_structure}  contains several results that shed light on the structure of \prog semidecidable sets and in Section \ref{sec_negative} we present the announced negative results. Finally, Section \ref{sec_future} contains a list of related problems for possible future work. 
\section{Background}\label{sec_background}
\subsection{Notations and basic definitions.}

We assume the reader is familiar with computability theory. Let $\{\varphi_{e}\}_{e\in\N}$ be an effective enumeration  of the set 
of computable partial functions. We denote by $\W$ the collection of c.e.\ subsets of $\N$ and $W_e=\dom(\varphi_{e})$ the induced effective enumerations of its elements. If $A\in \W$, an \defin{index} of $A$ is a number $e$ such that $W_e=A$. If $A$ is a c.e.\ set, implicitly given by an index, $A[s]$ is the finite subset of $A$ enumerated by stage $s$, so that $A[s]\subseteq A[s+1]$ and $A=\bigcup_s A[s]$. We use the notation $A[\at s]=A[s]\setminus A[s-1]$ if $s\geq 1$ and $A[\at 0]=A[0]$. If $F$ is a finite subset of $\N$ then $\up{F}$ is the collection of supersets of $F$.  $\BB=\N^{\N}$ will denote  Baire space.



\medskip
\subsection{Effective topological spaces.} 


An \defin{effective topological space} is a tuple $(\X,\tau,\B)$ where $(\X,\tau)$ is a non-empty topological space, $\B=\{\B_{i}\}_{i\in\N}$ is numbered basis such that there exists a computable function $f:\N\times \N\to\N$ satisfying $\B_i\cap \B_j=\bigcup_{k\in W_{f(i,j)}}\B_k$.

Given an effective topological space $\X$, the standard representation is defined as a surjective map $\rho:\dom(\rho)\subseteq\BB\to \X$ satisfying $\rho(f)=x$ whenever $\{f(n):n\in\N\}=\{i:x\in \B_i\}$. We will call any  $f\in \rho^{-1}(x)$, a \defin{\oracle name} of $x$. An element $x$ is \defin{computable} if it has a computable \oracle name.  We denote by $\bd{X_c}$ the set of computable points. The countable set $X_c$ has a canonical numbering $\nu$ defined by $\nu(i)=x$ if $\varphi_i$ is a name of $x$. We will call such an $e$ a \defin{\prog name} of $x$. To facilitate the reading of the paper, we will use the font $A, N, U$ when working on the space $X_{c}$, and the fonts $\A, \mathcal{N}, \U$ when working on $\X$.  

\subsubsection{\oracle computability and \prog computability}

Let $(\X,\tau,\B)$ and $(\Y,\tau',\B')$ be effective topological spaces. In what follows R stands for both \emph{\oracles{}} and \emph{\progs}.    A set $A\subseteq X_{c}$ is \defin{R-semidecidable} if there is a Turing machine $M$ which, when provided with an R-name of $x$, halts if and only if $x\in A$. A function $F:X_{c}\to Y_{c}$ is \defin{R-computable} if there is a Turing machine $M$ which,  when provided with an R-name of $x$,  writes  an R-name for $f(x)$ on its one-way output tape. It is not hard to see that a function $f:X_{c}\to Y_{c}$ is R-computable if and only if the sets $f^{-1}(\B'_i)$ are uniformly R-semidecidable.
 
 \begin{remark}
It is worth noting that for a function $f:X_{c}\to Y_{c}$, being \prog computable is equivalent to having a Machine $M$ which, provided with a  \prog name of $x$, outputs a \emph{\oracle name} of $f(x)$. Indeed, combining the program for $x$ with the program for $M$ gives a program for $f(x)$.
We also note that a function $f:X_{c}\to Y_{c}$ which is \oracle computable does not necessarily extends to a \oracle computable function $\overline{f}:\X\to\Y$. 

\end{remark}


 A numbering $\eta$ of $X_c$ is \defin{admissible} if it is equivalent to the canonical numbering $\nu$ in the sense that there exists  partial computable functions $f$  and $g$ such that $\nu=\eta\circ f$ on $\dom(\nu)$ and $\eta=\nu\circ g$ on $\dom(\eta)$.  The \prog computability notions do not depend on the choice of the admissible numbering.  We will often use the admissible numbering $\eta$ of $X_c$ defined by $\eta(e)=x$  whenever $W_e=\{i\in\N:x\in \B_i\}$.

\oracle computability and topology are closely related. A set $\U\subseteq \X$ is an \defin{effective open set} if there exists $e\in\N$ such that $\U= \bigcup_{i\in W_e}\B_i$.  If $A= \U \cap X_{c}$, we will then say that $A$ is effectively open \defin{in $\bd{X_{c}}$}. The connection is established by  the following result (see \cite{Wei00}).

\begin{theorem} A set  $A\subseteq X_{c}$  is  \oracle semidecidable  if and only if it is effectively open in $X_{c}$. Therefore, a function $f:X_{c}\to Y_{c}$ is \oracle computable if and only if it is effectively continuous, i.e.\ the sets $f^{-1}(\B'_i)$ are uniformly effectively open in $X_{c}$.
\end{theorem}

As mentioned in the introduction, in order to have an analogous result for \prog computability, we have to use Ershov's topology on $X_{c}$, which may be different from the topology of $\X$ restricted to $X_{c}$.


\begin{example}
Let $\BB=\N^{\N}$ be the Baire space. For each finite sequence $u$, let $[u]$ be the set of infinite extensions of $u$, called a cylinder. We endow $\BB$ with the topology generated by the cylinders, which is an effective topology. The standard numbering $\varphi_e$ of partial computable functions, restricted to the indices of total functions is an admissible numbering of $\BB_{c}$.
\end{example}

\begin{example}
Let $\P(\N)$ be space of subsets of $\N$. For each finite set $F\subseteq\N$, let $[F]$ be the set of supersets of $F$. We endow $\P(\N)$ with the Scott topology, generated by the sets $[F]$, which is an effective topology. The standard numbering $W_e=\dom(\varphi_e)$ of c.e.\ sets is admissible numbering of  $P(\N)_{c}$
\end{example}

\section{Main results}\label{main-results}

In this section, $(\X,\tau,\B,\nu_\B)$ is always an effective topological space and $X_c$ is the subset of computable elements. We start by explaining the main idea behind our results.  Let $x\in X_{c}$ be a fixed element. From a machine \oracle semideciding a set $A$ containing $x$, one can compute a neighborhood $\mathcal{N}$ of $x$ such that for every element $y\in X_{c}$ the following implication holds:
\begin{equation}\label{implication}
y\in \mathcal{N}\implies y\in A.
\end{equation}

Now assume that $A$ has the weaker property of being \prog semidecidable, and still contains $x$. From a machine \prog semideciding $A$ one cannot in general compute such a neighborhood, which may not exist as shown by Friedberg's example. However, from the \prog name of any other element $y\in X_{c}$ one can still compute a neighborhood $\mathcal{N}_{y}$ of $x$ such that implication \eqref{implication} holds. Further, as a finite intersection of neighborhoods is still a neighborhood, one can compute a neighborhood $\mathcal{N}$ satisfying implication \eqref{implication} for all $y$ in a given finite set. Using this argument we can show that the problem $x\in A$ can be \oracle semidecided \emph{as soon as we know, in addition, a finite list of programs containing at least one for $x$}. This additional information is equivalent to having any upper bound on the Kolmogorov complexity of $x$, which leads us to the notion of \K computability that we now introduce.

\subsection{\K computability}

\begin{definition}
The \defin{Kolmogorov complexity K(x)} of a computable element $x\in X_c$ is the length of a shortest program computing a \oracle name of $x$.
\end{definition}
In this paper, whether we use prefix-free, monotone or plain machines will not make any difference so we do not need to specify the definition any further.

\begin{definition}
A \defin{K-name} of a computable element $x\in X_c$ consists of a pair $(k,f)$ where $k\geq K(x)$ and $f$ is a \oracle name of $x$.
\end{definition}

\begin{remark} Note that $k$ is only an upper  bound on the Kolmogorov complexity of $x$ and not necessarily of $f$, which may even be non computable. Note also that knowing any such $k$ is effectively equivalent to knowing any upper bound on a \prog name of $x$. This is what we will rather use in our proofs.

\end{remark}

The \defin{\K computability} notions are defined in the same way as in the previous section. We will denote by $\bd{X_{c}(k)}$ the set of computable elements whose Kolmogorov complexity is at most $k$. Note that $X_{c}=\bigcup_{k}X_{c}(k)$ and that \K computability is the same as \oracle computability on $X_{c}(k)$, uniformly in $k$. In particular, a set $A\subseteq X_{c}$ is \K semidecidable iff there exists uniformly effective open sets $\U_{k}$ such that $A\cap X_{c}(k)=\U_{k}\cap X_{c}(k)$. 

Thus, for each notion of computability we have so far three  versions, depending on the way the objects are represented. 

It is clear that one can compute \K names from \prog names. An important first observation is the fact that the converse does not necessarily holds. In other words, the  representations underlying \prog computability and \K computability are not equivalent.

\begin{proposition}\label{prop_K_Markov_repr}
In general, it is not possible to compute \prog names from \K names. 
\end{proposition}

\begin{proof}
Let $f:\N\to\N$ be a computable function such that $x_{f(i)}=0^\omega$ if $\varphi_i(i)$ does not halt, $x_{f(i)}=0^t1^\omega$ if $\varphi_i(i)$ halts in time $t$.

Assume that there is an oracle Turing machine $M$ that converts each \K name of a sequence into a program computing the sequence (a \prog name).  There is a computable function $k:\N\to\N$ such that $k(i)$ is an upper bound on the complexity of $x_{f(i)}$ and of $0^\omega$. Given $i$, run $M$ on $0^{k(i)}10^\omega$, let $u_i$ be its use when it halts. $M$ halts and outputs a program computing the sequence $0^\omega$. On $0^{k(i)}1x_{f(i)}$ it must halt and output a program computing $x_{f(i)}$, so if $x_{f(i)}$ contains a $1$ then it must occur early in the sequence, namely before $u_i-k(i)-1$. As a result, if $\varphi_i(i)$ halts then it must halt in time $u_i-k(i)-1$. It enables one to decide the Halting problem, which is a contradiction.
\end{proof}

One can show that on Cantor space,  \prog names are limit-computable  (can be \emph{learned}) from \K names: given $x$ and $k\geq K(x)$, one can compute a sequence of natural numbers converging to an index of $x$ (this problem was investigated in the context of inductive inference \cite{FreiWie79}). 
A c.e.\ set, however, cannot be learned in this way.  Actually one can prove a stronger statement. 
\begin{proposition}\label{prop_K_Markov_inference}
There is no Turing functional $\Phi$ that, given an index $e$ and a \oracle name of a set $W$ which is either $\N$ or $W_e$, computes a sequence of numbers converging to an index of $W$. 
\end{proposition}

\begin{proof}
Assume that such a $\Phi$ exists. Using the Recursion theorem, we define an index $a$ in the following way. At the same time we enumerate a c.e.\ set $W_a$ and we build an oracle $f$. 

We define $W_a=\bigcup_i F_i$ where $F_i$ is a computable sequence of finite sets and $F_i\subseteq F_{i+1}$. At the same time we define an oracle $t$ as the limit of a computable sequence of finite strings $t_i$ such that $t_{i+1}$ extends $t_i$. The finite string $t_i$ contains exactly the elements of $F_i$.

The strings $t_i$ are such that $\Phi^{t_i}(a)$ outputs a (finite or infinite) sequence of indices taking at least $i$ values. Hence $\Phi^t(a)$ outputs an infinite sequence of indices taking infinitely many values.

We start with $F_0=\emptyset$ and $t_0$ is the empty string. Assume $F_i$ and $t_i$ have been defined. Look for a finite extension $u$ of $t_i$ such that $\Phi^u(a)$ eventually outputs an index $i$ such that $W_i$ contains some number that is not in $u$ (such a $u$ must be found: on a representation of $\N$ starting with $t_i$, $\Phi$ must eventually output a program enumerating $\N$, let $u$ be the finite part of the oracle that is read when such a program is output). Let $F_{i+1}$ be the set of elements enumerated in $u$. Run $\Phi^{u0^\omega}(a)$ and look for a number $j\neq i$ output later than $i$ (such a $j$ must be found, otherwise the oracle is a representation of $W_a$ so it must eventually stabilize on a program $j$ enumerating the elements of $u$, and $j$ must be different from $i$). Let $t_{i+1}$ be the part of the oracle that is used in the computation of $j$.
\end{proof}

The rest of this  section is devoted to show that, despite the facts above, the notions of  \prog computability and \K computability are indeed equivalent to a large extent.

\subsection{Equivalence between \prog computability and \K computability}

We will use the Recursion Theorem. See \cite{Rog87}.

\begin{theorem}[Recursion Theorem]
For every computable total function  $f$, there exists  $e$ such that $\varphi_{e}=\varphi_{f(e)}$. Moreover, $p$ can be computed from an index of $f$. 

\end{theorem}

The following Lemma contains the main technical  arguments.

\begin{lemma}\label{lem_ext}
Let $A$ be a c.e.\ subset of $\N$. There exist uniformly effective Scott open sets $\U_k\subseteq \P(\N)$, such that for every c.e.\ set $E$ the following hold:
\begin{enumerate}
\item[i)] if all the indices of $E$ belong to $A$ then $E\in \U_k$ for every $k$,
\item[ii)] if no index of $E$ belongs to $A$ then $E\notin \U_k$ for every $k\geq K(E)$.
\end{enumerate}
\end{lemma}

The argument is uniform: the open sets $\U_k$ are effective, uniformly in a c.e.\ index of $A$.
\begin{proof}
Using the Recursion theorem, there is a computable function $e(a,b)$ such that for all $a,b\in\N$,
\begin{equation*}
W_{e(a,b)}=\begin{cases}
W_a&\text{if $e(a,b)\notin A$,}\\
W_a[t]\cup W_b&\text{if $e(a,b)\in A[\at t]$.}
\end{cases}
\end{equation*}

Let $k\in\N$. We define an effective open set $\U_k$.  Compute $b_{k}$ such every element whose complexity is less than $k$ has an index less than $b_{k}$.  If $a$ is such that for all $b\leq b_k$, $e(a,b)\in A$ then let $t$ be minimal such that $e(a,b)\in A[t]$ for all $b\leq b_k$, enumerate $\up{W_a[t]}$ into $\U_k$.

We now check the two announced conditions. \emph{i)} Let $E\subseteq\N$ be a c.e.\ set.
Assume that every index of $E$ belongs to $A$ and let $a$ be an index of $E$. For all $b$, $e(a,b)\in A$ (otherwise $e(a,b)$ is an index of $W_a=E$ but $e(a,b)\notin A$, contradiction), so $\U_k$ contains $\up{W_a[t]}$ for some $t$, which contains $E$.
\emph{ii)} Assume that $K(E)\leq k$, that no index of $E$ belongs to $A$ and that $E\in \U_k$.  Let $b\leq b_k$ be an index of $E$. As $E\in \U_k$, $E$ belongs to some $[W_a[t]]$ enumerated into $\U_k$ (here $a$ is not the same as above and is not assumed to be an index of $E$). As $e(a,b)\in A$, $W_{e(a,b)}=W_a[t']\cup W_b$ for some $t'\leq t$. As $W_a[t']\subseteq W_a[t]\subseteq W_b$, $e(a,b)$ is an index of $E$ that belongs to $A$, contradicting the assumption.
\end{proof}


We know state the main explicit versions of Theorems {A} and B. 

\begin{theorem}\label{thm_semidec}
Let $\X$ be an effective topological space. A set $A\subseteq X_c$ is \prog semidecidable iff  it is \K semidecidable. The equivalence is uniform.
\end{theorem}
\begin{proof}
Every effective topological space is \oracle computably homeomorphic to a subspace of $\P(\N)$: to $x\in \X$, associate $\{i\in\N:x\in \B_i\}$ where $\B_i$ is enumeration of the basis of $\X$. Hence we can assume that $\X$ is a subspace of $\P(\N)$. Let $I\subseteq\N$ be a c.e.\ set such that for all $e\in\N$ for which $W_e\in X_{c}$, it holds $W_e\in A\iff e\in I$. Each c.e.\ set $E\in X_{c}$ either has all its indices in $I$ or has no index in $I$, so the effective open sets $\U_k$ provided by Lemma \ref{lem_ext} coincide with $A$ on the set of elements of $X_{c}$ whose complexity is at most $k$. Now, a machine \K semideciding $A$ works as follows: given a \oracle name of  $E\in X_{c}$ and $k\geq K(E)$, it tests whether $E\in \U_k$ and halts in this case only.
\end{proof}

\begin{corollary}\label{cor_function}
Let $\X,\Y$ be effective topological spaces. A function $f:X_c\to Y_{c}$ is \prog computable iff $f$ is \K computable. The equivalence is uniform.
\end{corollary}
\begin{proof}
Let $B_i$ be the numbered basis of $\Y$. $f$ is \prog computable iff the sets $f^{-1}(\B_i)$ are uniformly \prog semidecidable iff these sets are \K semidecidable (Theorem \ref{thm_semidec}) iff $f$ is \K computable. 
\end{proof}



We now show that the argument in the proof of Lemma  \ref{lem_ext} can be extended from  semidecidability to  weaker classes of properties, showing that for most algorithmic tasks, the additional information given by programs is indeed just an upper bound on the Kolmogorov complexity. 

\noindent\textbf{Hierarchies.} Let $\X$ be an effective topological space. We consider the finite levels of the effective Borel hierarchy, defined as follows. The class $\Sigma^0_1$ consists of the effective open sets. The class $\Sigma^0_{n+1}$ consists of the effective unions of differences of $\Sigma^0_n$-sets. The classes $\Pi^0_n$ consists of complements of $\Sigma^0_n$-sets. The class $\Delta^0_n$ is the intersection of $\Sigma^0_n$ and $\Pi^0_n$. Inside the class $\Delta^0_2$ we consider the finite levels of the effective difference hierarchy. For $n\in\N$, the class $\mathcal{D}_n$ consists of the differences of $n$ effective open sets, i.e.\ the sets $(\U_0\setminus \U_1)\cup \ldots (\U_{n-2}\setminus \U_{n-1})$ if $n$ is even and the sets $(\U_0\setminus \U_1)\cup \ldots \U_{n-1}$ if $n$ is odd. In the case  $\X=\N$ with the discrete topology, the effective Borel hierarchy is exactly the arithmetical hierarchy, the class $\mathcal{D}_n$ of effective difference hierarchy is exactly the class of $n$-c.e.\ sets.

\begin{theorem}\label{thm_nce}
A set $A\subseteq X_c$ is \prog $n$-c.e.\ iff it is \K $n$-c.e. More precisely, the set of indices of elements of $A$ is $n$-c.e.\ on the set of indices of $X_c$ iff there exist uniformly effective open sets $\U^1_k,\ldots,\U^n_k$ such that $A\cap X_c(k)=D_n(\U^1_k,\ldots,\U^n_k)\cap X_c(k)$.
\end{theorem}

\begin{proof}
Again we can assume w.l.o.g.\ that $\X$ is a subspace of $\P(\N)$. Let $A_0\supseteq A_1\supseteq\ldots \supseteq A_{n-1}$ such that if $W_e\in X_c$ then $W_e\in A\iff e\in A_0\setminus A_1\cup A_2\setminus A_3\ldots$. We denote a tuple $(a_0,\ldots,a_n)\in\N^{n+1}$ by $\overline{a}$. There is a computable function $e(\overline{a})$ such that
\begin{equation*}
W_{e(\overline{a})}=\begin{cases}
W_{a_0}&\text{if $e(\overline{a})\notin A_0$,}\\
W_{a_0}[t_0]\cup W_{a_1}&\text{if $e(\overline{a})\in A_0[\at t_0]\setminus A_1$,}\\
W_{a_0}[t_0]\cup W_{a_1}[t_1]\cup W_{a_2}&\text{if $e(\overline{a})\in A_0[\at t_0]\cap A_1[\at t_1]\setminus A_2$,}\\
\ldots&\\
W_{a_0}[t_0]\cup W_{a_1}[t_1]\cup \ldots\cup W_{a_{n-1}}&\text{if $e(\overline{a})\in A_0[\at t_0]\cap \ldots \cap A_{n-2}[\at t_{n-2}]\setminus A_{n-1}$,}\\
W_{a_0}[t_0]\cup W_{a_1}[t_1]\cup \ldots\cup W_{a_{n}}&\text{if $e(\overline{a})\in A_0[\at t_0]\cap \ldots \cap A_{n-2}[\at t_{n-2}]\cap A_{n-1}[\at t_{n-1}]$.}
\end{cases}
\end{equation*}

Given $k$, let $b_k$ be un upper bound on the indices of elements whose Kolmogorov complexity is at most $k$. If $a_0$ is such that for all $a_1,\ldots,a_{n}\leq b_k$, $e(a_0,\ldots,a_{n})\in A_0$ then let $t_{0}$ be minimal such that all these numbers belong to $A_0[t_{0}]$, enumerate $\up{W_{a_0}[t_0]}$ in $U_0$. By induction, let $1\leq i<n$ and assume $a_0,\ldots,a_{i-1}$ have been accepted with $t_0,\ldots,t_{i-1}$. If $a_i$ is such that for all $a_{i+1},\ldots,a_{n}\leq b_k$, $e(a_0,\ldots,a_{n})\in A_i$ then let $t_i$ be minimal such that all these numbers belong to $A_i[t_i]$ and enumerate $\up{W_{a_0}[t_0]\cup\ldots\cup W_{a_i}[t_i]}$ in $U_i$.

In the same way, one checks that if $W_e\in X_c(k)$ then $W_e\in A\iff W_e\in U_0\setminus U_1\cup U_2\setminus U_3\cup \ldots$.
\end{proof}

It is known from \cite{Selivanov84} that there exists a \prog $2$-c.e.\ subset of $\P(\N)$ that is not even $\Pi^0_2$. Hence \prog $2$-c.e.\ sets are not the differences of \prog semidecidable sets.

In the following theorem, we need to assume an additional property on the space $\X$. Namely, that the domain of the standard representation on $\X$ is a $\Pi^0_2$ set.  This is case for example for the so called \emph{quasi-Polish spaces} (see \cite{Brecht13}).  

\begin{theorem}\label{thm_sigma2}
A set $A\subseteq X_c$ is \prog $\Sigma^0_2$ iff it is \K $\Sigma^0_2$. More precisely,  the set of indices of elements of $A$ is $\Sigma^0_2$ on the set of indices of $X_c$ iff there exist uniformly effective open sets $\U^n_k, \V^n_k$ such that $A\cap X_c(k)=\bigcup_n (\U^n_k\setminus \V^n_k)\cap X_c(k)$.
\end{theorem}
\begin{remark} In case $\X$ is a Polish space, the sets $\V_{k}$ are not needed and therefore the last part of the statement reads $A\cap X_c(k)=\bigcup_n \U^n_k\cap X_c(k)$.
\end{remark}

\begin{proof}[Proof of Theorem \ref{thm_sigma2}]
We show that if $A$ is \prog $\Pi^0_2$ then $A$ is \K $\Pi^0_2$, which is equivalent to the statement by replacing $A$ with its complement. We use the numbering $x_e=x$ if $W_e$ is the set of indices of basic neighborhoods of $x$. Let $P=\bigcap_n P_n\subseteq\N$ be $\Pi^0_2$ ($P_n$ are uniformly c.e.) such that if $e$ is an index of $x\in X_c$ then $x\in A\iff e\in P$. The assumption about the space implies that the set of indices of elements of $X_c$ is a $\Pi^0_2$-set $Q=\bigcap_n Q_n\subseteq\N$, where $Q_n$ are uniformly c.e.\ sets.

Given $i$, let $C_i=\{x_i\}$ if $i$ is an index and $i\in P\cap Q$, $C_i=\emptyset$ otherwise. $C_i$ is $\Pi^0_2$, uniformly in $i$. Indeed, for each $n$, define the uniformly effective open sets
\begin{equation*}
(\U_{2n},\V_{2n})=\begin{cases}
(B_n,\emptyset) &\text{if $n\notin W_i$,}\\
(X,B_n) & \text{if $n\in W_i$},
\end{cases}
\end{equation*}
and
\begin{equation*}
(\U_{2n+1},\V_{2n+1})=\begin{cases}
(X,\emptyset) &\text{if $i\notin P_n\cap Q_n$,}\\
(X,X) & \text{if $i\in P_n\cap Q_n$}.
\end{cases}
\end{equation*}
One has $C_i=\bigcap_n(\U_n\setminus \V_n)^c$.

Given $k\in\N$, compute $b_k$ such that every element of Kolmogorov complexity at most $k$ has an index $\leq b_k$. The set $\bigcup_{i\leq b_k}C_i$ is $\Pi^0_2$, uniformly in $k$. This set is exactly $A\cap X_c(k)$.
\end{proof}


\section{Structure of \prog semidecidable sets}\label{sec_structure}

Here we provide several results that shed light on the computational and topological structure of \prog semidecidable sets. Our first result shows that \prog semidecidable sets share some of the nice properties of \oracle semidecidable sets.

\begin{proposition}\label{prop_dense}
Assume that $\X$ contains a dense computable sequence. Given a \prog semidecidable set $A$, it is semi-decidable whether $A$ is non-empty. If $A$ is non-empty, one can compute a sequence of points $\{x_{i}\} \subseteq A$ which is dense in $A$. 
\end{proposition}
\begin{proof}
Using the Recursion theorem, there is a computable function $e(a)$ such that $x_{e(a)}=x_a$ if $e(a)\notin A$,  or $x_{e(a)}$ is some point from the dense sequence in some neighborhood of $x_a$ if $e(a)\in A$ at time $t$. $A$ is non-empty iff there is $a$ such that $e(a)\in A$. 
When $A$ is non-empty, one can compute an element in $A$: look for $a$ such that $e(a)\in A$, $x_{e(a)}$ is such a point. To get a computable dense sequence, apply this argument in parallel to the intersection of $A$ with each basic open set $\B_{i}$.
\end{proof}

The following result provides an upper bound on the effective Borel complexity of \prog semidecidable sets. 
\begin{proposition}\label{prop_pi02}
Let $A\subseteq X_{c}$ be \prog semidecidable. There exist uniformly effective open sets $\U_k\subseteq\X$ such that $A=\bigcap_k\U_k\cap X_c$.
\end{proposition}
\begin{proof}
Let $\U_k$ be the effective open sets from the proof of Theorem \ref{thm_semidec} and define $\S=\bigcap_k \U_k $. We already know that $A\subseteq \U_k$ for all $k$. If $x\in X_{c}\cap \S$ then let $k\geq K(x)$. Since $x\in X_{c}(k)\cap \U_k=X_{c}(k)\cap A$, we conclude that $x\in A$.
\end{proof}

The result above is actually tight, as the following theorem shows. For a finite string $u$, let us define the monotone complexity $\Km(u)$ of $u$ as the length of a shortest program computing a (finite or infinite) binary sequence extending $u$. The program writes its output on a one-way output tape and may never halt. Again the precise definition of $\Km(u)$ (Levin or Schnorr monotone or process complexity) does not make any difference for our purposes. The only important property is that for a computable sequence $x$, $\Km(x\res{n})\leq Km(x)$ for all $n$. For the seek of completeness, let us recall original Friedberg's example. We present it in a way that is more convenient for our purposes.  

\begin{theorem}[Friedberg]
On the Cantor space, the set 
\begin{equation*}
\A=\{0^\omega\}\cup\bigcup_{n:Km(0^n1) < \log(n) - 1}[0^n1].
\end{equation*}
is  \prog semidecidable but not open.  Hence the Ershov topology is strictly stronger than the Cantor topology. 
\end{theorem}
\begin{proof} We show that $\A$ is \K semidecidable. Given an infinite binary sequence $x$ (a \oracle description) and a bound $k$ on $K(x)$, we only need to read the first $e=2^{k+2}$ bits of $x$. If we see only zeros, we accept. Otherwise one gets $0^n1\ldots$ for some $n<e$, then test whether $Km(0^n1)< \log(n)-1$. 
\end{proof}

\begin{remark}
Friedberg's example happens to be $\Sigma^0_2$. It is an effective open set appended with a limit point.  We strengthen  Friedberg's example by constructing a \prog semidecidable set which is far from being open. 
\end{remark}

\begin{theorem}\label{thm_notsigma2}
There is a \prog semidecidable subset of $\cantor_{c}$ that is not $\Sigma^0_2$. It is a non-empty closed subset of $\cantor_{c}$ with empty interior, defined by
\begin{equation*}
A=\{x\in\cantor_{c}:\forall n, \Km(x\res{n})< n/2+c\}\quad\text{ for some sufficiently large $c\in\N$.}
\end{equation*}
\end{theorem}

\begin{proof}
We choose $c$ such that for some computable sequence $x$, $K(x)\leq c$, hence $A$ is non-empty as it contains $x$. We first show that $A$ is \K semidecidable. First, the function $u\mapsto \Km(u)$ is right-c.e. Now, given $x$ and some $k\geq K(x)$, $x\in A$ iff for all $n\leq 2(k-c)$, $\Km(x\res{n})< n/2+c$, as for all $n>2(k-c)$, $\Km(x\res{n})\leq K(x)\leq k<n/2+c$.

Here we denote $\cantor$ by $\X$ and the set of computable sequences by $X_c$. $A$ is a subset of $X_c$. We show that there is no $\Sigma^0_2$-subset of $\X$ whose intersection with $X_c$ is $A$. Let $\overline{A}$ be the closure of $A$ in $\X$ (it might not be $\{x\in\X:\forall n, \Km(x\res{n})< n/2+c\}$). Here is the argument:

\begin{enumerate}
\item $A$ has empty interior in $X_c$, i.e.\ there is no cylinder $[u]$ such that $[u]\cap X_c\subseteq A$. Indeed, given a finite string $u$ and a sufficiently large $k$, for most words $v$ of length $k$, $\Km(uv)\geq |uv|/2+c$ so $[uv]$ is disjoint from $A$.
\item If $\P\subseteq\X$ is a $\Pi^0_1$-set and $\P\cap X_c\subseteq A$ then $\P$ is nowhere dense in $A$. Indeed, if there exists a cylinder $[u]$ such that $\emptyset\neq A\cap [u]\subseteq \P$ then $A\cap [u]=\P\cap [u]\cap X_c$ is both \prog semidecidable and \prog co-semidecidable hence by Kreisel-Lacombe-Sh\oe nfield/Ceitin theorem it is clopen on $X_c$, so $A$ has non-empty interior in $X_c$, contradicting the first point.
\item By Proposition \ref{prop_dense}, $\overline{A}$ is a c.e.\ closed subset of $\X$ (it contains a dense computable sequence) hence a $\Pi^0_2$-set. Let $\S\subseteq\X$ be the $\Pi^0_2$-set given by Proposition \ref{prop_pi02}, satisfying $A=\S\cap X_c$. Let $\S'=\overline{A}\cap \S$. $\S'$ is a $\Pi^0_2$-set which contains a dense computable sequence, and $A$ is exactly the set of computable elements of $\S'$. From this it follows that computable Baire theorem holds on $\S'$: if the sets $\P_i$ are uniformly $\Pi^0_1$-sets that have empty interior in $\S'$ then one can compute some $x$ in $\S'\setminus \bigcup_i\P_i$.
\item Now, if $\P_i$ are uniformly $\Pi^0_1$-sets such that each $\P_i\cap X_c$ is contained in $A$ then by the second point $\P_i$ has empty interior in $A$, and also in $\S'\subseteq A$, so by the third point one can compute some $x$ in $\S'\setminus \bigcup_i\P_i$. As $x$ is computable and belongs to $\S'$, $x$ belongs to $A$ so $\bigcup_i \P_i$ does not cover $A$.\qedhere
\end{enumerate}
\end{proof}

For the following results, we restrict our attention to the space $\X=\overline{\N}=\N\cup\{\infty\}$ whose topology is generated by the singletons $\{n\}$  and the semi-lines $[n,\infty]$, for $n\in \N$. Note that every point in this space is computable, so that $\X = X_{c}$.

Friedberg's example translated to this space reads  $\{x\in\overline{\N}: K(x) < \log(x) - 1\}$, which inspires the following definition.

\begin{definition}
We define the  \defin{Friedberg sets} of $\overline{\N}$ to be the ones of the form $\{x\in\overline{\N}: K(x) < h(x)\}$, where $h:\N\to\N$ is any computable order,  namely, any non decreasing unbounded computable function. 
\end{definition}
 Note that a computable order can always be extended to a computable function $h:\overline{\N}\to\overline{\N}$, with $h(\infty)=\infty$.
 
  Friedberg sets are \prog semidecidable just like the original set. The next two results show that, unlike Cantor space, the only \prog semidecidable sets over $\overline{\N}$ which are not \oracle semidecidable are essentially the  Friedberg sets. 

\begin{proposition}\label{friedberg}
If $A\subseteq\overline{\N}$ is \prog semidecidable and contains $\infty$ then there is a computable order $h$ such that $A$ contains a Friedberg set. \end{proposition}

\begin{proof}
Since $A$ is \K semidecidable, for each $k$ one can compute $p(k)$ such that $[p(k),\infty]\cap\{x:K(x) \leq  k\}\subseteq A$. One can assume that $p(k)$ is increasing.  Let $h(n)=\min\{i:p(i)> n\}$.  If $n\notin A$ then $p(K(n))>n$ (just take $k=K(n)$), so one has $h(n) \leq K(n)$.
\end{proof}

\begin{remark}Observe that $K(x)$ here coincides with the usual notion of Kolmogorov complexity of natural numbers. 
\end{remark}

Proposition \ref{friedberg} provides a nice characterization of the Ershov's open sets. 

\begin{corollary}
The Ershov topology is generated by the singletons $\{n\}$ and the Friedberg sets. 
\end{corollary}

\begin{remark}Note that the sets $[n,\infty]$ can be expressed as the Friedberg sets  $\{x\in\overline{\N}:K(x)<h(x)\}$ where $h(x)=0$ for $x<n$ and $h(x)=c(x+1)$ for $x\geq n$, where $c$ is such that $K(n)\leq c(n+1)$ for all $n\in\N$.
\end{remark}

Whether or not one can find such a characterization on other spaces such as the Cantor space is an interesting question.

We end this section by observing that, unlike \oracle semidecidable sets, \prog semidecidable sets cannot be effectively enumerated. 

\begin{proposition}\label{prop_enum}
There is no effective enumeration of the \prog semidecidable subsets of $\overline{\N}$.
\end{proposition}
\begin{proof}
Let $A_i$ be a sequence of uniformly \prog semidecidable sets, coming with uniformly c.e.\ sets $E_i\subseteq\N$ such that $E_i\cap\dom(\nu)=\nu^{-1}(A_i)$. One can extract the sets that contain $\infty$ (let $e_0$ be some index of $\infty$, one can enumerate the numbers $i$ such that $E_i$ contains $e_0$), so we can assume w.l.o.g.\ that each $A_i$ contains $\infty$. For each $i$ one can compute an increasing computable function $f_{i}:\N\to\N$ whose range is contained in $A_i$. Now we build a \prog semidecidable set $A$ that contains $\infty$ and differs from each $A_i$. Let $f$ be a computable order such that for each $i$ and all sufficiently large $k$, $f_i(k)<f(k)$ (for instance, $f(k)=\max(f_0(k),\ldots,f_k(k))+1$). Here we use another version of Kolmogorov complexity: $C(x)$ is the minimal index of $x$. We now define $A=\{x\in\overline{\N}:f(C(x))\leq x\}$. $A$ is \prog semidecidable by the usual argument.

We show that $A$ differs from each $A_i$. Let $i\in\N$. As $C\circ f_i$ is one-to-one, there exist infinitely many $k\in\N$ such that $C(f_i(k))\geq k$. Moreover if $k$ is sufficiently large then $f_i(C(f_i(k)))<f(C(f_i(k)))$. Hence there exists $k$ such that $f_i(k)\leq f_i(C(f_i(k)))<f(C(f_i(k))$. Let $x=f_i(k)$: $x\in A_i$ by construction of $f_i$ and $x<f(C(x))$ so $x\notin A$.
\end{proof}

\section{When Markov beats Kolmogorov}\label{sec_negative}

In this section we explore the limits of our methods. We first look at the relativized case, and show that there are simple cases that separate \prog  computability from  \K computability.  However, we also show that, interestingly, the equivalence persists if the space has a Polish structure. 
 
\subsection{Relativization}

Let $\sierp=\{\bot,\top\}$ be the Sierpi\'nski space with topology given by $\{\emptyset, \{\top\}, \{\bot,\top\}\}$.  Note that as $\sierp$ is finite, \K computability is trivially equivalent to \oracle computability simply because all the elements share a common upper bound on their Kolmogorov complexities, which therefore provides no interesting information. Relativizing w.r.t. the Halting set, we can then separate \prog decidability from \oracle decidability, and therefore from \K decidability. 

\begin{remark}\label{prop_sierp_semidec_relative}
 The set $\{\bot\}\subseteq\sierp$ is \prog decidable relative to the halting set but is not \oracle decidable relative to any oracle.
\end{remark}
\begin{proof}
It is not decidable relative to any oracle simply because it is not clopen.
\end{proof}

Similarly, using $\emptyset''$ we can separate, over $\P(\N)$, \K semidecidability from \oracle semidecidability (without oracle, the two notions coincide with \prog decidability by Rice-Shapiro theorem).

\begin{proposition}\label{prop_zero_second} The set $\{\N\}\subseteq \P(\N)$ is \K semidecidable relative to $\emptyset''$ but is not \oracle semidecidable relative to any oracle. 
\end{proposition}

\begin{proof}
Let $E\subseteq\N$ be a c.e.\ set and $k$ an upper bound on its Kolmogorov complexity. From $k$ we know that $A$ has an index in some finite set $F$. Using $\emptyset''$ we, for each $e\in F$, decide whether $W_e=\N$ and compute, when $W_e\neq\N$, an element outside $W_e$. We then wait that each one of this finite set of elements appears in $A$, enumerated on the input tape, and accept $A$ if it is the case.
\end{proof}

However, metric spaces behave differently. Although stated on Cantor space, the next result extends to any computable metric space \cite{Wei00}. 

\begin{proposition}\label{prop_cantor_semidec_relative}
Let $A$ be an oracle computing the Halting set. A subset of $\cantor_c$ is \prog semidecidable relative to $A$ if and only if it is \K semidecidable relative to $A$. 
\end{proposition}
\begin{proof}
Assume that $\A$ is \prog semidecidable relative to the halting set, via some oracle Turing machine $M$. Assume we are given the halting set, $x\in\cantor_c$ and an upper bound $k$ on its Kolmogorov complexity. Look for a partition of $\{0,\ldots,k\}$ into three sets $A,B,C$ such that:
\begin{itemize}
\item for every $a\in A$, $x_a$ is incompatible with $x$ (there exists $n$ such that $x_a(n)\down\neq x(n)$),
\item for every $b\in B$, $x_b$ is partial (there exists $n$ such that $x_b(n)\dv$),
\item for every $c\in C$, $c$ is accepted by $M$ with the halting set as oracle.
\end{itemize}
Once such a partition is found, accept $x$.

First, the tests can be effectively done relative to the halting set: incompatibility of $x_a$ with $x$ is semidecidable, partiality of $x_b$ is semidecidable relative to the halting set. If such a partition is found then $x$ must have an index in $C$ so $x$ must satisfy the property. Conversely if $x$ satisfies the property then for every $i\leq k$, either $i$ is an index of $x$ or $x_i$ is incompatible with $x$ or $x_i$ is partial, so a partition exists and will be eventually found.
\end{proof}

\subsection{Functions to non-effective topological spaces}

In this section we provide results that strictly separate our three notions: \prog computability, \K computability and \oracle computability. The idea of the constructions is to build uniform versions of the examples given in Remark \ref{prop_sierp_semidec_relative} and Proposition \ref{prop_zero_second}. For this, one can imagine a function with two arguments, where the second argument $A\in \BB$ is always provided to the machine by a \oracle name and plays the role of the oracle.  The only difficulty is then to make the computation work in a uniform way in the oracle parameter. In order to get a well defined function w.r.t.\ our models, we express it as a function of the first argument only, but with values on $\O(\BB)$, which is the set of open subsets of the Baire space endowed with the topology generated by the following sets: given a compact set $K\subseteq\BB$, the class of open subsets of $\BB$ containing $K$ is open. This topology is not countably-based, and hence it is not an effective topology.  However it does have an admissible representation \cite{Schroder02}.

We now present the details of the simplest case, a uniform version of Remark \ref{prop_sierp_semidec_relative}. This result contrasts with Corollary \ref{cor_function}.

\begin{theorem}\label{thm_sierp_OB}
There exists a \prog computable function $F:\sierp\to\O(\BB)$ that is not \K computable.
\end{theorem}
\begin{proof}
We use the admissible numbering $\nu_\sierp$ of $\sierp$ defined by $\nu_\sierp(e)=\top$ if $\varphi_e(e)\down$, $\nu_\sierp(e)=\bot$ otherwise.
We define two effective open sets $U_\bot,U_\top$ and define $F(\bot)=U_\bot$ and $F(\top)=U_\top$. First, let $U_\bot=\BB$.
Let $T:\N\to\N$ be defined as follows: $T(n)$ is the halting time of $\varphi_n(n)$ if it halts, $T(n)=0$ otherwise. The open set $U_\top:=\BB\setminus \{T\}$ happens to be effective.
First the function $F$ is not Lacombe computable because it is not continuous: indeed, $F$ is not monotonic as $\bot\leq \top$ but $U_\bot=\BB$ is not contained in $U_\top\subsetneq\BB$. As $\sierp$ is finite, $F$ is not \K computable neither.
However $F$ is \prog computable. Given an index $e$ of $s\in\sierp$, enumerate $U_\top$ and enumerate the set of functions $f$ such that $\varphi_e(e)$ does not halt in exactly $f(e)$ steps. The latter set of functions is effectively open, uniformly in $e$. If $\varphi_e(e)\dv$ then the whole space $\BB$ is enumerated. If $\varphi_e(e)\down$ then nothing more than $U_\top$ is enumerated.
Intuitively, given $e$ and $f$, from $T$ one can decide whether $\varphi_e(e)$ halts, i.e.\ whether $\nu_\sierp(e)=\bot$.
\end{proof}

A similar construction, based on Proposition \ref{prop_zero_second}, yields a function $F:\P(\N) \to \O(\BB)$ which is \K computable but not \oracle computable by replacing the function $T$ from Theorem \ref{thm_sierp_OB} by a function $T'$ computing $\emptyset''$ and such that $\BB\setminus \{T'\}$ is effectively open.

Combining all these results, and using that fact that  Theorem \ref{thm_sierp_OB} can clearly also be realized using $\P(\N)$ in place of $\sierp$,   we obtain our announced Theorem C. 

\begin{theorem}\label{separation}For functions from $\P(\N)$ with values on $\O(\BB)$ on has that:
\begin{equation*}
\emph{\prog computability \,>\, \K computability \, > \, \oracle computability.} 
\end{equation*}
\end{theorem}

While \oracle computable functions are always Scott continuous (i.e.\ monotone and compact), one can show that \K computable functions are always monotone but not necessarily compact. \prog computable functions may even not be monotone.

Let us now briefly discuss whether Theorem \ref{separation} holds for functions from the Cantor space to $\O(\BB)$. Friedberg's example of a \progs{} (hence \Ks)-semidecidable set that is not \oracle semidecidable directly implies the second inequality. However the idea behind the proof of the first inequality cannot be applied on Cantor space. Indeed, using Proposition \ref{prop_cantor_semidec_relative} one can show that the analog of the function of Theorem \ref{thm_sierp_OB} is actually \K computable.

\begin{proposition}
The function $G:\cantor\to\O(\BB)$ mapping $0^\omega$ to $\BB$ and any other sequence to $\BB\setminus \{T\}$ is \K computable.
\end{proposition}
\begin{proof}
Given $x,k$ and $f$, apply the algorithm given by Proposition \ref{prop_cantor_semidec_relative} to semi-decide, if $f=T$, whether $x=0^\omega$. In parallel, semidecide whether $f\neq T$.
\end{proof}

%

We leave the following question open: is there a \prog computable function from the Cantor space to $\O(\BB)$ that is not \K computable? A simpler version of this question: is there an oracle that separates \K semidecidability from \prog semidecidability on Cantor space? We note that, by Proposition \ref{prop_cantor_semidec_relative}, such an oracle should not compute $\emptyset'$.

\section{Future work}\label{sec_future}
We list a few problems for future work.
\begin{itemize}
\item Find a characterization of the Ershov topology on other spaces than $\Nbar$, like the Cantor space.
\item Determine for which levels of the effective difference hierarchy the \prog model and the \K model are equivalent. We know from Theorem \ref{thm_nce} that the equivalence holds for the \emph{finite} levels. What about the level $\omega$?
\item All our results hold when the space $\X$ is an effective topological space. However the three models also make sense on any represented space. It seems like an interesting research program to study the extent to which our results are valid in this case.
\item Compare the effective Borel hierarchy induced by the \prog semidecidable sets, the hierarchy induced by the arithmetical hierarchy on the indices and the effective Borel hierarchy induced by the standard topology.
\end{itemize}


\begin{thebibliography}{KLS57}

\bibitem[Cei62]{Ceitin62}
G.~S. Ceitin.
\newblock Algorithmic operators in constructive metric spaces.
\newblock {\em Trudy Matematiki Instituta Steklov}, 67:295--361, 1962.
\newblock English translation: {\it American Mathematical Society
  Translations}, series 2, 64:1-80, 1967.

\bibitem[dB13]{Brecht13}
Matthew de~Brecht.
\newblock Quasi-polish spaces.
\newblock {\em Ann. Pure Appl. Logic}, 164(3):356--381, 2013.

\bibitem[Fri58]{Friedberg58}
Richard~M. Friedberg.
\newblock Un contre-exemple relatif aux fonctionnelles r\'ecursives.
\newblock {\em Comptes Rendus de l'Acad\'emie des Sciences}, 247:852--854,
  1958.

\bibitem[FW79]{FreiWie79}
Rusins Freivalds and Rolf Wiehagen.
\newblock Inductive inference with additional information.
\newblock {\em Journal of Information Processing and Cybernetics}, 15:179--185,
  1979.

\bibitem[Grz57]{Grzegorczyk57}
Andrzej Grzegorczyk.
\newblock On the definitions of computable real continuous functions.
\newblock {\em Fundamenta Mathematicae}, 44:61--71, 1957.

\bibitem[Her96]{Hertling96}
Peter Hertling.
\newblock Computable real functions: {T}ype 1 computability versus {T}ype 2
  computability.
\newblock In {\em {CCA}}, 1996.

\bibitem[KLS57]{KreLacSch57}
G.~Kreisel, D.~Lacombe, and J.R. {Sch\oe nfield}.
\newblock Fonctionnelles r\'ecursivement d\'efinissables et fonctionnelles
  r\'ecursives.
\newblock {\em Comptes Rendus de l'Acad\'emie des Sciences}, 245:399--402,
  1957.

\bibitem[Kus06]{Kushner06}
Boris~A. Kushner.
\newblock The constructive mathematics of A. A. Markov.
\newblock {\em Amer. Math. Monthly}, 113(6):559--566, 2006.

\bibitem[Lac55]{Lacombe55}
Daniel Lacombe.
\newblock Extension de la notion de fonction r\'ecursive aux fonctions d?une ou
  plusieurs variables r\'eelles {I-III}.
\newblock {\em Comptes Rendus {A}cad\'emie des {S}ciences {P}aris},
  240,241:2478--2480,13--14,151--153, 1955.

\bibitem[Mar54]{Markov54}
A.~A. Markov.
\newblock On the continuity of constructive functions (russian).
\newblock {\em Uspekhi Mat. Nauk}, 9:226--230, 1954.

\bibitem[MS55]{MyhillShe55}
J.~Myhill and J.~C. Shepherdson.
\newblock Effective operations on partial recursive functions.
\newblock {\em Mathematical Logic Quarterly}, 1(4):310--317, 1955.

\bibitem[PE60]{PourEl60}
Marian~B. Pour-El.
\newblock A comparison of five ``computable'' operators.
\newblock {\em Mathematical Logic Quarterly}, 6(15-22):325--340, 1960.

\bibitem[Ric53]{Rice53}
H.~G. Rice.
\newblock Classes of recursively enumerable sets and their decision problems.
\newblock {\em Transactions of the American Mathematical Society}, 74(2):pp.
  358--366, 1953.

\bibitem[Rog87]{Rog87}
Hartley~Jr. Rogers.
\newblock {\em Theory of Recursive Functions and Effective Computability}.
\newblock MIT Press, Cambridge, MA, USA, 1987.

\bibitem[Sch02]{Schroder02}
Matthias Schr{\"o}der.
\newblock Extended admissibility.
\newblock {\em Theoretical Computer Science}, 284(2):519--538, 2002.

\bibitem[Sel84]{Selivanov84}
Victor~L. Selivanov.
\newblock {Index sets in the hyperarithmetical hierarchy}.
\newblock {\em Siberian Mathematical Journal}, 25:474--488, 1984.

\bibitem[Sha56]{Shapiro56}
N.~Shapiro.
\newblock Degrees of computability.
\newblock {\em Transactions of the American Mathematical Society}, 82:281--299,
  1956.

\bibitem[Spr01]{Spreen01}
Dieter Spreen.
\newblock Representations versus numberings: on the relationship of two
  computability notions.
\newblock {\em Theoretical Computer Science}, 262(1):473--499, 2001.

\bibitem[Tur36]{Tur36}
Alan Turing.
\newblock On computable numbers, with an application to the
  Entscheidungsproblem.
\newblock {\em Proceedings of the London Mathematical Society}, 2, 42:230--265,
  1936.

\bibitem[Wei00]{Wei00}
Klaus Weihrauch.
\newblock {\em Computable Analysis}.
\newblock Springer, Berlin, 2000.

\end{thebibliography}

\end{document}